\documentclass[twocolumns]{article}

\usepackage{amsmath,amssymb}
\usepackage[standard]{ntheorem}
\usepackage{graphicx} 
\usepackage{color}
\usepackage{dsfont}
\usepackage{graphicx}
\usepackage{subfigure}
\usepackage{stmaryrd}
\usepackage{color}
\usepackage{url}

\begin{document}

\title{Steganography: a class of secure and robust algorithms}

\author{Jacques M. Bahi, Jean-Fran\c{c}ois Couchot, and Christophe Guyeux\thanks{Authors in alphabetic order}\\University of Franche-Comt\'{e}, Computer Science Laboratory, Belfort, France\\ \{jacques.bahi, jean-francois.couchot, christophe.guyeux\}@univ-fcomte.fr}



\newcommand{\Nats}[0]{\ensuremath{\mathbb{N}}}
\newcommand{\Z}[0]{\ensuremath{\mathbb{Z}}}
\newcommand{\R}[0]{\ensuremath{\mathbb{R}}}
\newcommand{\Bool}[0]{\ensuremath{\mathds{B}}}
\newcommand{\StratSet}[0]{\ensuremath{\mathbb{S}}}

\maketitle

\begin{abstract}
This research work presents a new class of non-blind information hiding
algorithms that are stego-secure and robust.
They are based on some finite domains iterations having the Devaney's 
topological chaos property.
Thanks to a complete formalization of the approach we prove   
security against watermark-only attacks
of a large class of steganographic algorithms.
Finally a complete study of robustness is given in frequency DWT and DCT domains.
\end{abstract}

\section{Introduction}\label{sec:intro}

This work focus on non-blind  binary information hiding chaotic schemes: 
the  original host  is  required  to  extract the  binary hidden
information. This  context is indeed not
as restrictive as it could primarily appear.
Firstly, it allows to prove 
the authenticity of a document sent  through the  Internet 
(the original document is stored whereas the stego content is sent). 
Secondly, Alice and Bob can establish an hidden channel into a
streaming video 
(Alice and Bob both have the same movie, and  Alice hide information into 
the frame number $k$  iff the binary digit
number $k$ of its hidden message is  1).
Thirdly, based on a similar idea, a same
given image can be marked  several times by using various secret parameters
owned both by Alice and Bob. Thus more than one bit can be embedded into a given
image  by using  this work. Lastly,  non-blind watermarking  is
useful in  network's anonymity and intrusion  detection \cite{Houmansadr09}, and
to protect digital data sending through the Internet \cite{P1150442004}.
Furthermore,  enlarging the given payload  of a data  hiding scheme leads
clearly to  a degradation of its  security: the smallest the  number of embedded
bits is, the better the security is.

Chaos-based approaches are frequently proposed to improve
the quality of schemes in information 
hiding~\cite{Wu2007,Liu07,CongJQZ06,Zhu06}.
In these works, the understanding of chaotic systems
is almost intuitive: a kind of noise-like spread system
with sensitive dependence on initial condition.
Practically, some well-known chaotic maps are used
either in the data encryption stage~\cite{Liu07,CongJQZ06}, 
in the embedding into the carrier medium,
or in both~\cite{Wu2007,Wu2007bis}.
Methods referenced above are almost based on two
fundamental chaotic maps, namely the Chebychev and logistic maps, which range in $\mathbb{R}$. 
To avoid justifying that functions which are chaotic in $\mathbb{R}$
still remain chaotic in the computing representation (\textit{i.e.},
floating numbers) we argue that functions should be iterated on finite domains.
Boolean discrete-time dynamical systems (BS) are thus iterated.

Furthermore, previously referenced works often focus on  
discretion and/or robustness properties, but they do not consider security.  
As far as we know, stego-security~\cite{Cayre2008} and chaos-security 
have only been proven 
on the spread spectrum watermarking~\cite{Cox97securespread},
and on
the dhCI algorithm~\cite{gfb10:ip}, which is notably based on iterating  
the negation function.
We argue that other functions can provide algorithms as secure as the dhCI one.
This work generalizes thus this latter algorithm and formalizes all
its stages. Due to this formalization, we address the proofs of the two
security properties for a large class of steganography approaches.

This research work is organized as follows.
It firstly introduces the new class of algorithms (Sec.~\ref{sec:formalization}), 
 which is the first contribution.
Next, the Section~\ref{sec:security} presents a 
State-of-the-art  in information hiding security and shows how secure  
is our approach. The proof is the second contribution.
The chaos-security property is studied in Sec.~\ref{sec:chaossecurity}
and instances of algorithms guaranteeing that desired property are presented.
This is the fourth contribution. 
Applications in frequency domains (namely DWT and DCT embedding) 
are formalized and corresponding experiments are 
given in Sec.~\ref{sec:applications}. 
This shows the applicability of the whole approach.
Finally, conclusive remarks and perspectives are given in Sec.~\ref{sec:concl}.

\section{Information hiding algorithm: formalization}
\label{sec:formalization}

As far as we know, no result rules that the chaotic behavior of a function 
that has been established on $\mathbb{R}$ remains on the floating numbers.
As stated before, this work presents the alternative to iterate a Boolean map:
results that are theoretically obtained in that domain are preserved  
during implementations.
In this section, we first give some recalls on Boolean discrete
dynamical Systems (BS). 
With this material, next sections 
formalize the information hiding algorithms based on 
these Boolean iterations.

\subsection{Boolean discrete dynamical systems}\label{sub:bdds}

Let us denote by $\llbracket a ; b \rrbracket$ the interval of integers:
$\{a, a+1, \hdots, b\}$, where $a \leqslant b$.

Let $n$ be a positive integer. A Boolean discrete-time 
network is a discrete dynamical
system defined from a {\emph{Boolean map}}
$f:\Bool^n\to\Bool^n$ s.t. 
\[
  x=(x_1,\dots,x_n)\mapsto f(x)=(f_1(x),\dots,f_n(x)),
\]
{\emph{and an iteration scheme}}: parallel, serial,
asynchronous\ldots 
With the parallel iteration scheme, 
the dynamics of the system are described by $x^{t+1}=f(x^t)$
where $x^0 \in \Bool^n$.
Let thus $F_f: \llbracket1;n\rrbracket\times \Bool^{n}$ to $\Bool^n$ 
be defined by
\[
F_f(i,x)=(x_1,\dots,x_{i-1},f_i(x),x_{i+1},\dots,x_n),
\]
with the \emph{asynchronous} scheme,
the dynamics of the system are described by $x^{t+1}=F_f(s^t,x^t)$
where $x^0\in\Bool^n$ and $s$ is a {\emph{strategy}}, \textit{i.e.}, a sequence 
in $\llbracket1;n\rrbracket^\Nats$.
Notice that this scheme only modifies one element at each iteration.

Let $G_f$ be the map from $\mathcal{X}= \llbracket1;n\rrbracket^\Nats\times\Bool^n$ to 
itself s.t.
\[
G_f(s,x)=(\sigma(s),F_f(s^0,x)),
\] 
where $\sigma(s)^t=s^{t+1}$ for all $t$ in $\Nats$. 
Notice that parallel iteration of $G_f$ from an initial point
$X^0=(s,x^0)$ describes the ``same dynamics'' as the asynchronous
iteration of $f$ induced by the initial point $x^0$ and the strategy
$s$.

Finally, let $f$ be a map from $\Bool^n$ to itself. The
{\emph{asynchronous iteration graph}} associated with $f$ is the
directed graph $\Gamma(f)$ defined by: the set of vertices is
$\Bool^n$; for all $x\in\Bool^n$ and $i\in \llbracket1;n\rrbracket$,
$\Gamma(f)$ contains an arc from $x$ to $F_f(i,x)$.

We have already established~\cite{GuyeuxThese10} that we can define a
distance $d$ on $\mathcal{X}$ such that 
$G_f$ is a continuous and chaotic function according to 
Devaney~\cite{Devaney}.
The next section focus on  the coding step of  the steganographic algorithm
based on $G_f$ iterations. 

\subsection{Coding}\label{sub:wmcoding}
In what follows, $y$ always stands  
for a digital content we wish to hide into a digital host $x$.

The data hiding scheme presented here does not constrain media to have 
a constant size. It is indeed sufficient to provide a function and a strategy 
that may be  parametrized with the size of the elements to modify. 
The \emph{mode} and the \emph{strategy-adapter} defined below achieve 
this goal.  

\begin{definition}[Mode]
\label{def:mode}
A map $f$, which associates to any $n \in \mathds{N}$ an application 
$f_n : \mathds{B}^n \rightarrow \mathds{B}^n$, is called a \emph{mode}.
\end{definition}



For instance, the \emph{negation mode} is defined by the map that
assigns to every integer $n \in \mathds{N}^*$ the function 
$${\neg}_n:\mathds{B}^n \to \mathds{B}^n, 
(x_1, \hdots, x_n) \mapsto (\overline{x_1}, \hdots, \overline{x_n}).$$

\begin{definition}[Strategy-Adapter]
  \label{def:strategy-adapter}
  A \emph{strategy-adapter}\index{configuration} is a function $\mathcal{S}$ 
  from $\Nats$ to the set of integer sequences, 
  which associates to $n$ a sequence 
  $S \in  \llbracket 1, n\rrbracket^\mathds{N}$.
\end{definition}

Intuitively, a strategy-adapter aims at generating a strategy 
$(S^t)^{t \in \Nats}$ where each term $S^t$ belongs to 
$\llbracket 1, n \rrbracket$. Moreover it may be parametrized
in order to depend on digital media to embed. 

For instance, let us define the  \emph{Chaotic Iterations with Independent Strategy}
(\emph{CIIS}) strategy-adapter.
The CIIS strategy-adapter with parameters 
$(K,y,\alpha,l) \in [0,1]\times [0,1] \times ]0, 0.5[ \times \mathds{N}$
is the function that associates to any  $n \in \Nats$ the sequence
$(S^t)^{t \in \mathds{N}}$ defined by:

 \begin{itemize} 
 \item $K^0 = \textit{bin}(y) \oplus \textit{bin}(K)$: $K^0$ is the real number whose binary decomposition is equal to the bitwise exclusive or (xor)
   between the binary decompositions of $y$ and of  $K$;
 \item $\forall t \leqslant l, K^{t+1} = F(K^t,\alpha)$;
 \item $\forall t \leqslant l, S^t = \left \lfloor n \times K^t \right \rfloor + 1$;
 \item $\forall t > l, S^t = 0$.
 \end{itemize}
where $F$ is the
piecewise linear chaotic map~\cite{Shujun1}, 
recalled in what follows:

\begin{definition}[Piecewise linear chaotic map]
 \label{def:fonction chaotique linéaire par morceaux}
Let $\alpha \in ]0; 0.5[$ be a control parameter. 
The \emph{piecewise linear chaotic map} is the map $F$
defined by: 
 $$
 F(t,\alpha) = \left\{
 \begin{array}{cl} 
 \dfrac{t}{\alpha} & t \in [0; \alpha],\\ 
 \dfrac{t-\alpha}{\frac{1}{2}-\alpha} & t \in [\alpha;\frac{1}{2}],\\ 
 F(1-t,\alpha) & t \in [\frac{1}{2}; 1].\\ 
 \end{array}
 \right. 
 $$
 \end{definition}

Contrary to the logistic map, the use of this piecewise linear chaotic map 
is relevant in cryptographic usages \cite{Arroyo08}.

Parameters of CIIS strategy-adapter will be instantiate as follows: 
$K$ is the secret embedding key, $y$ is the secret message, 
$\alpha$ is the threshold of the piecewise linear chaotic map,
which can be set as $K$ or can act as a second secret key.
Lastly, $l$ is for the iteration number bound:
enlarging its value improve the chaotic behavior of the scheme,
but the time required to achieve the embedding grows too.

Another strategy-adapter is the 
\emph{Chaotic Iterations with Dependent Strategy} (CIDS) 
with parameters $(l,X) \in \mathds{N}\times \mathds{B}^\mathds{N}$, 
which is the function that maps any $ n \in \mathds{N}$ to
the sequence $\left(S^t\right)^{t \in \mathds{N}}$ defined by:
\begin{itemize}
\item $\forall t \leqslant l$, if $t \leqslant l$ and $X^t = 1$, 
  then $S^t=t$, else $S^t=1$;
\item $\forall t > l, S^t = 0$.
\end{itemize}

Let us notice that the terms of $x$ that may be replaced by terms taken
from $y$ are less important than other: they could be changed 
without be perceived as such. More generally, a 
\emph{signification function} 
attaches a weight to each term defining a digital media,
w.r.t. its position $t$:

\begin{definition}[Signification function]
A \emph{signification function} is a real sequence 
$(u^k)^{k \in \Nats}$. 
\end{definition}





For instance, let us consider a set of    
grayscale images stored into portable graymap format (P3-PGM):
each pixel ranges between 256 gray levels, \textit{i.e.},
is memorized with eight bits.
In that context, we consider 
$u^k = 8 - (k  \mod  8)$  to be the $k$-th term of a signification function 
$(u^k)^{k \in \Nats}$. 
Intuitively, in each group of eight bits (\textit{i.e.}, for each pixel) 
the first bit has an importance equal to 8, whereas the last bit has an
importance equal to 1. This is compliant with the idea that
changing the first bit affects more the image than changing the last one.

\begin{definition}[Significance of coefficients]\label{def:msc,lsc}
Let $(u^k)^{k \in \Nats}$ be a signification function, 
$m$ and $M$ be two reals s.t. $m < M$. Then 
the \emph{most significant coefficients (MSCs)} of $x$ is the finite 
  vector $u_M$, 
the \emph{least significant coefficients (LSCs)} of $x$ is the 
finite vector $u_m$, and 
the \emph{passive coefficients} of $x$ is the finite vector $u_p$ such that:
\begin{eqnarray*}
  u_M &=& \left( k ~ \big|~ k \in \mathds{N} \textrm{ and } u^k 
    \geqslant M \textrm{ and }  k \le \mid x \mid \right) \\
  u_m &=& \left( k ~ \big|~ k \in \mathds{N} \textrm{ and } u^k 
  \le m \textrm{ and }  k \le \mid x \mid \right) \\
   u_p &=& \left( k ~ \big|~ k \in \mathds{N} \textrm{ and } 
u^k \in ]m;M[ \textrm{ and }  k \le \mid x \mid \right)
\end{eqnarray*}
 \end{definition}

For a given host content $x$,
MSCs are then ranks of $x$  that describe the relevant part
of the image, whereas LSCs translate its less significant parts.
We are then ready to decompose an host $x$ into its coefficients and 
then to recompose it. Next definitions formalize these two steps. 

\begin{definition}[Decomposition function]
Let $(u^k)^{k \in \Nats}$ be a signification function, 
$\mathfrak{B}$ the set of finite binary sequences,
$\mathfrak{N}$ the set of finite integer sequences, 
$m$ and $M$ be two reals s.t. $m < M$.  
Any host $x$ can be decomposed into 
\[
(u_M,u_m,u_p,\phi_{M},\phi_{m},\phi_{p})
\in
\mathfrak{N} \times 
\mathfrak{N} \times 
\mathfrak{N} \times 
\mathfrak{B} \times 
\mathfrak{B} \times 
\mathfrak{B} 
\]
where
\begin{itemize}
\item $u_M$, $u_m$, and $u_p$ are coefficients defined in Definition  
\ref{def:msc,lsc};
\item $\phi_{M} = \left( x^{u^1_M}, x^{u^2_M}, \ldots,x^{u^{|u_M|}_M}\right)$;
 \item $\phi_{m} = \left( x^{u^1_m}, x^{u^2_m}, \ldots,x^{u^{|u_m|}_m} \right)$;
 \item $\phi_{p} =\left( x^{u^1_p}, x^{u^2_p}, \ldots,x^{u^{|u_p|}_p}\right) $.
 \end{itemize}
The function that associates the decomposed host to any digital host is 
the \emph{decomposition function}. It is 
further referred as $\textit{dec}(u,m,M)$ since it is parametrized by 
$u$, $m$, and $M$. Notice that $u$ is a shortcut for $(u^k)^{k \in \Nats}$.
\end{definition}

\begin{definition}[Recomposition]
Let 
$(u_M,u_m,u_p,\phi_{M},\phi_{m},\phi_{p}) \in 
\mathfrak{N} \times 
\mathfrak{N} \times 
\mathfrak{N} \times 
\mathfrak{B} \times 
\mathfrak{B} \times 
\mathfrak{B} 
$ s.t.
\begin{itemize}
\item the sets of elements in $u_M$, elements in $u_m$, and 
elements in $u_p$ are a partition of $\llbracket 1, n\rrbracket$;
\item $|u_M| = |\varphi_M|$, $|u_m| = |\varphi_m|$, and $|u_p| = |\varphi_p|$.  
\end{itemize}
One can associate the vector 
\[
x = 
\sum_{i=1}^{|u_M|} \varphi^i_M . e_{{u^i_M}} +  
\sum_{i=1}^{|u_m|} \varphi^i_m .e_{{u^i_m}} +  
\sum_{i=1}^{|u_p|} \varphi^i_p. e_{{u^i_p}} 
\]
\noindent where 
$(e_i)_{i \in \mathds{N}}$ is the usual basis of the $\mathds{R}-$vectorial space $\left(\mathds{R}^\mathds{N}, +, .\right)$ (that is to say, $e_i^j = \delta_{ij}$, where $\delta_{ij}$ is the Kronecker symbol).
The function that associates $x$ to any 
$(u_M,u_m,u_p,\phi_{M},\phi_{m},\phi_{p})$ following the above constraints 
is called the \emph{recomposition function}.
\end{definition}

The embedding consists in the replacement of the values of 
$\phi_{m}$ of $x$'s LSCs  by $y$. 
It then composes the two decomposition and
recomposition functions seen previously. More formally:

\begin{definition}[Embedding media]
Let $\textit{dec}(u,m,M)$ be a decomposition function,
$x$ be a host content,
$(u_M,u_m,u_p,\phi_{M},\phi_{m},\phi_{p})$ be its image by $\textit{dec}(u,m,M)$, 
and $y$ be a digital media of size $|u_m|$.
The digital media $z$ resulting on the embedding of $y$ into $x$ is 
%
the image of $(u_M,u_m,u_p,\phi_{M},y,\phi_{p})$
by  the recomposition function $\textit{rec}$.
\end{definition}

Let us then define the dhCI information hiding scheme
presented in~\cite{gfb10:ip}:

\begin{definition}[Data hiding dhCI]
 \label{def:dhCI}
Let $\textit{dec}(u,m,M)$ be a decomposition function,
$f$ be a mode, 
$\mathcal{S}$ be a strategy adapter,
$x$ be an host content,\linebreak
$(u_M,u_m,u_p,\phi_{M},\phi_{m},\phi_{p})$ 
be its image by $\textit{dec}(u,m,M)$,
$q$ be a positive natural number,  
and $y$ be a digital media of size $l=|u_m|$.

The dhCI dissimulation  maps any
$(x,y)$  to the digital media $z$ resulting on the embedding of
$\hat{y}$ into $x$, s.t.

\begin{itemize}
\item we instantiate the mode $f$ with parameter $l=|u_m|$, leading to 
  the function $f_{l}:\Bool^{l} \rightarrow \Bool^{l}$;
\item we instantiate the strategy adapter $\mathcal{S}$ 
with parameter $y$ (and possibly some other ones);
this instantiation leads to the strategy $S_y \in \llbracket 1;l\rrbracket ^{\Nats}$.

\item we iterate $G_{f_l}$ with initial configuration $(S_y,\phi_{m})$;
\item $\hat{y}$ is finally the $q$-th term of these iterations.
\end{itemize}
\end{definition}

To summarize, iterations are realized on the LSCs of the
host content
(the mode gives the iterate function,  
the strategy-adapter gives its strategy), 
and the last computed configuration is re-injected into the host content, 
in place of the former LSCs.

Notice that in order to preserve the unpredictable behavior of the system, 
the size of the digital medias is not fixed.
This approach is thus self adapted to any media, and more particularly to
any size of LSCs. 
However this flexibility enlarges the complexity of the presentation: 
we had to give Definitions~\ref{def:mode} and~\ref{def:strategy-adapter} 
respectively of mode and strategy adapter.

\begin{figure}[ht]
\centering
\includegraphics[width=8.5cm]{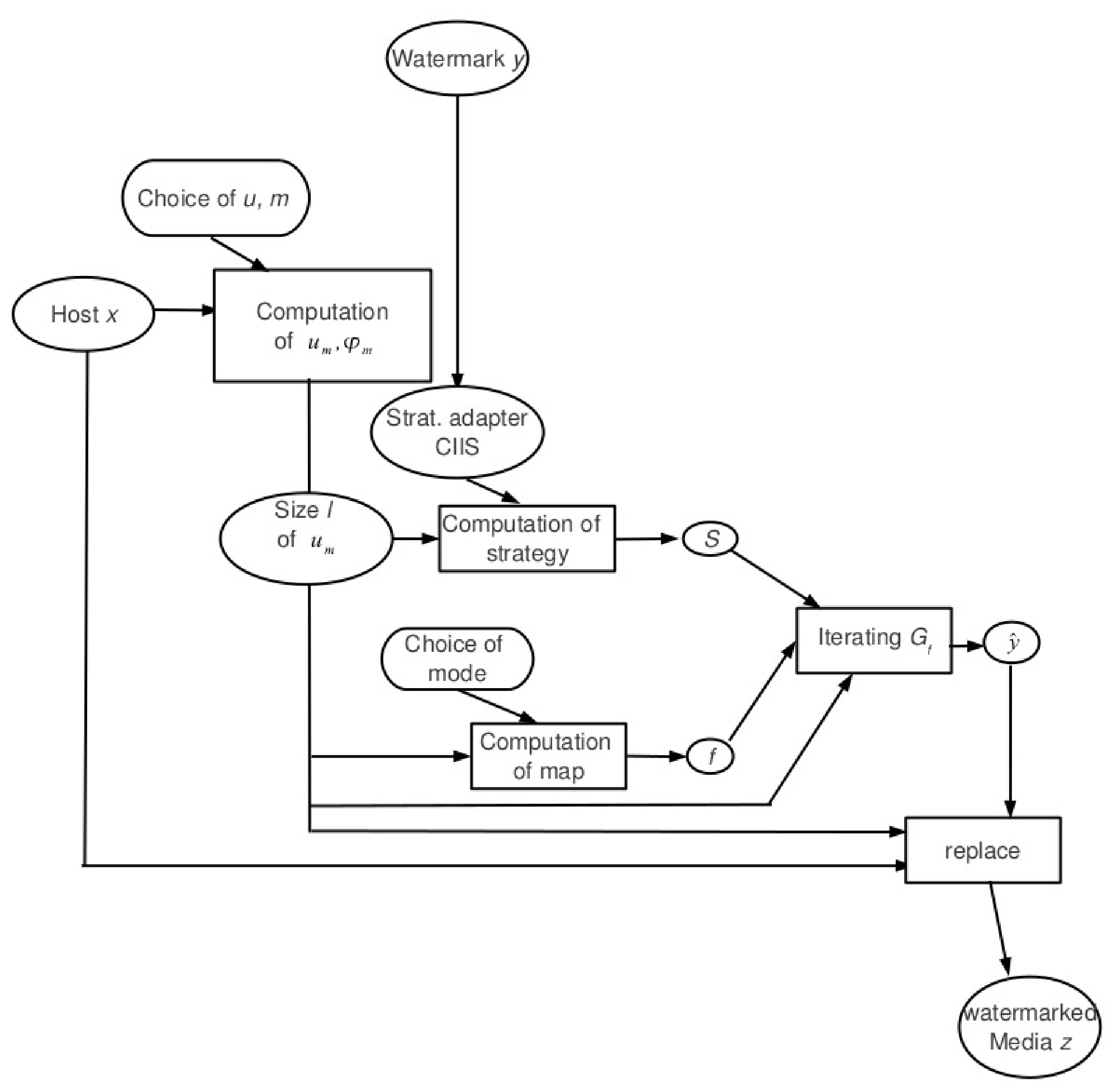}
\caption{The dhCI dissimulation scheme}
\label{fig:organigramme}
\end{figure}

Next section shows how to check whether a media contains a watermark.

\subsection{Decoding}\label{sub:wmdecoding}

Let us firstly show how to formally check whether a given digital media $z$ 
results from the dissimulation of $y$ into the digital media $x$.

\begin{definition}[Watermarked content]
Let $\textit{dec}(u,m,M)$ be a decomposition function,
$f$ be a mode, 
$\mathcal{S}$ be a strategy adapter, 
$q$ be a positive natural number,  
$y$ be a digital media, and 
$(u_M,u_m,u_p,\phi_{M},\phi_{m},\phi_{p})$ be the 
image by $\textit{dec}(u,m,M)$  of  a digital media $x$. 
Then $z$ is \emph{watermarked} with $y$ if
the image by $\textit{dec}(u,m,M)$ of $z$ is 
$(u_M,u_m,u_p,\phi_{M},\hat{y},\phi_{p})$, where 
$\hat{y}$ is the right member of $G_{f_l}^q(S_y,\phi_{m})$.
\end{definition}

Various decision strategies are obviously  possible to determine whether a given
image $z$ is watermarked or not, depending  on the eventuality
that the considered image may have  been attacked.
For example, a  similarity percentage between $x$
and  $z$ can  be  computed and compared  to a  given
threshold. Other  possibilities are the use of  ROC curves or
the definition of a null hypothesis problem.

The next section recalls some security properties and shows how the 
\emph{dhCI dissimulation} algorithm verifies them.

\section{Security analysis}\label{sec:security}
\subsection{State-of-the-art in information hiding security}\label{sub:art}

As far as we know, Cachin~\cite{Cachin2004}
produces the first fundamental work in information hiding security:
in the context of steganography, the attempt of an attacker to distinguish 
between an innocent image and a stego-content is viewed as an hypothesis 
testing problem.
Mittelholzer~\cite{Mittelholzer99} next proposed the first theoretical 
framework for analyzing the security of a watermarking scheme.
Clarification between  robustness and security 
and classifications of watermarking attacks
have been firstly presented by Kalker~\cite{Kalker2001}.
This work has been deepened by Furon \emph{et al.}~\cite{Furon2002}, who have translated Kerckhoffs' principle (Alice and Bob shall only rely on some previously shared secret for privacy), from cryptography to data hiding. 

More recently~\cite{Cayre2005,Perez06} classified the information hiding  
attacks into categories, according to the type of information the attacker (Eve)
has access to:
\begin{itemize}
\item in Watermarked Only Attack (WOA) she only knows embedded contents $z$;
\item in Known Message Attack (KMA) she knows pairs $(z,y)$ of embedded
  contents and corresponding messages;
\item in Known Original Attack (KOA) she knows several pairs $(z,x)$ 
  of embedded contents and their corresponding original versions;
\item in Constant-Message Attack (CMA) she observes several embedded
  contents $z^1$,\ldots,$z^k$ and only knows that the unknown 
  hidden message $y$ is the same in all contents.
\end{itemize}

To the best of our knowledge, 
KMA, KOA, and CMA have not already been studied
due to the lack of theoretical framework.
In the opposite, security of data hiding against WOA can be evaluated,
by using a probabilistic approach recalled below.

\subsection{Stego-security}\label{sub:stegosecurity}

In the Simmons' prisoner problem~\cite{Simmons83}, Alice and Bob are in jail and
they want to,  possibly, devise an escape plan by  exchanging hidden messages in
innocent-looking  cover contents.  These  messages  are to  be  conveyed to  one
another by a common warden named Eve, who eavesdrops all contents and can choose
to interrupt the communication if they appear to be stego-contents.

Stego-security,  defined in  this well-known  context, is  the  highest security
class in Watermark-Only  Attack setup, which occurs when Eve  has only access to
several marked contents~\cite{Cayre2008}.

Let $\mathds{K}$ be the set of embedding keys, $p(X)$ the probabilistic model of
$N_0$ initial  host contents,  and $p(Y|K)$ the  probabilistic model  of $N_0$
marked contents s.t. each host  content has  been marked
with the same key $K$ and the same embedding function.

\begin{definition}[Stego-Security~\cite{Cayre2008}]
\label{Def:Stego-security}  The embedding  function  is \emph{stego-secure}
if  $\forall K \in \mathds{K}, p(Y|K)=p(X)$ is established.
\end{definition}



 Stego-security  states that  the knowledge  of  $K$ does  not help  to make  the
 difference  between $p(X)$ and  $p(Y)$.  This  definition implies  the following
 property:
 $$p(Y|K_1)= \cdots = p(Y|K_{N_k})=p(Y)=p(X)$$ 
 This property is equivalent to  a zero Kullback-Leibler divergence, which is the
 accepted definition of the "perfect secrecy" in steganography~\cite{Cachin2004}.

\subsection{The negation mode is stego-secure}
To make this article self-contained, this section recalls theorems and proofs of stego-security for negation mode published in~\cite{gfb10:ip}.

\begin{proposition} \emph{dhCI dissimulation}  of Definition \ref{def:dhCI} with
negation mode and  CIIS strategy-adapter is stego-secure, whereas  it is not the
case when using CIDS strategy-adapter.
\end{proposition}

\begin{proof}   On   the    one   hand,   let   us   suppose    that   $X   \sim
\mathbf{U}\left(\mathbb{B}^n\right)$  when  using  \linebreak CIIS$(K,\_,\_,l)$.
We  prove  by  a
mathematical   induction   that   $\forall    t   \in   \mathds{N},   X^t   \sim
\mathbf{U}\left(\mathbb{B}^n\right)$.

The     base     case     is     immediate,     as     $X^0     =     X     \sim
\mathbf{U}\left(\mathbb{B}^n\right)$. Let us now suppose that the statement $X^t
\sim  \mathbf{U}\left(\mathbb{B}^n\right)$  holds  until for  some $t$. 
Let  $e  \in
\mathbb{B}^n$   and   \linebreak   $\mathbf{B}_k=(0,\cdots,0,1,0,\cdots,0)   \in
\mathbb{B}^n$ (the digit $1$ is in position $k$).

So    
$P\left(X^{t+1}=e\right)=\sum_{k=1}^n
P\left(X^t=e\oplus\mathbf{B}_k,S^t=k\right)$ where  
$\oplus$ is again the bitwise exclusive or. 
These  two events are  independent when
using CIIS strategy-adapter 
(contrary to CIDS, CIIS is not built by using $X$),
 thus:
$$P\left(X^{t+1}=e\right)=\sum_{k=1}^n
P\left(X^t=e\oplus\mathbf{B}_k\right) \times  P\left(S^t=k\right).$$ 

According to the
inductive    hypothesis:   $P\left(X^{n+1}=e\right)=\frac{1}{2^n}   \sum_{k=1}^n
P\left(S^t=k\right)$.  The set  of events $\left \{ S^t=k \right  \}$ for $k \in
\llbracket  1;n \rrbracket$  is  a partition  of  the universe  of possible,  so
$\sum_{k=1}^n                  P\left(S^t=k\right)=1$.                  Finally,
$P\left(X^{t+1}=e\right)=\frac{1}{2^n}$,   which    leads   to   $X^{t+1}   \sim
\mathbf{U}\left(\mathbb{B}^n\right)$.   This  result  is  true  for all  $t  \in
\mathds{N}$ and then for $t=l$.

Since $P(Y|K)$ is $P(X^l)$ that is proven to be equal to $P(X)$,
we thus  have established that, 
$$\forall K  \in [0;1], P(Y|K)=P(X^{l})=P(X).$$ 
So   dhCI   dissimulation   with   CIIS
strategy-adapter is stego-secure.

On  the  other  hand,  due  to  the  definition  of  CIDS,  we  have  \linebreak
$P(Y=(1,1,\cdots,1)|K)=0$. 
So   there  is   no  uniform  repartition   for  the stego-contents $Y|K$.
\end{proof}

To sum up, Alice  and Bob can counteract Eve's attacks in  WOA setup, when using
dhCI dissimulation with  CIIS strategy-adapter.  To our best  knowledge, this is
the second time an information hiding scheme has been proven to be stego-secure:
the   former  was   the  spread-spectrum   technique  in   natural  marking
configuration with $\eta$ parameter equal to 1 \cite{Cayre2008}.

\subsection{A new class of $\varepsilon$-stego-secure schemes}

Let us prove that,
\begin{theorem}\label{th:stego}
Let $\epsilon$ be positive,
$l$ be any size of LSCs, 
$X   \sim \mathbf{U}\left(\mathbb{B}^l\right)$,
$f_l$ be an image mode s.t. 
$\Gamma(f_l)$ is strongly connected and 
the Markov matrix associated to $f_l$ 
is doubly stochastic. 
In the instantiated \emph{dhCI dissimulation} algorithm 
with any uniformly distributed (u.d.) strategy-adapter 
that is independent from $X$,  
there exists some positive natural number $q$ s.t.
$|p(X^q)- p(X)| < \epsilon$. 
\end{theorem}

\begin{proof}   
Let $\textit{deci}$ be the bijection between $\Bool^{l}$ and 
$\llbracket 0, 2^l-1 \rrbracket$ that associates the decimal value
of any  binary number in $\Bool^{l}$.
The probability $p(X^t) = (p(X^t= e_0),\dots,p(X^t= e_{2^l-1}))$ for $e_j \in \Bool^{l}$ is thus equal to 
$(p(\textit{deci}(X^t)= 0,\dots,p(\textit{deci}(X^t)= 2^l-1))$ further denoted by $\pi^t$.
Let $i \in \llbracket 0, 2^l -1 \rrbracket$, 
the probability $p(\textit{deci}(X^{t+1})= i)$  is 
\[
 \sum\limits^{2^l-1}_{j=0}  
\sum\limits^{l}_{k=1} 
p(\textit{deci}(X^{t}) = j , S^t = k , i =_k j , f_k(j) = i_k ) 
\]
\noindent 
where $ i =_k j $ is true iff the binary representations of 
$i$ and $j$ may only differ for the  $k$-th element,
and where 
$i_k$ abusively denotes, in this proof, the $k$-th element of the binary representation of 
$i$.

Next, due to the proposition's hypotheses on the strategy,
$p(\textit{deci}(X^t) = j , S^t = k , i =_k j, f_k(j) = i_k )$ is equal to  
$\frac{1}{l}.p(\textit{deci}(X^t) = j ,  i =_k j, f_k(j) = i_k)$.
Finally, since $i =_k j$ and $f_k(j) = i_k$ are constant during the 
iterative process  and thus does not depend on $X^t$, we have 
\[
\pi^{t+1}_i = \sum\limits^{2^l-1}_{j=0}
\pi^t_j.\frac{1}{l}  
\sum\limits^{l}_{k=1} 
p(i =_k j, f_k(j) = i_k ).
\]

Since 
$\frac{1}{l}  
\sum\limits^{l}_{k=1} 
p(i =_k j, f_k(j) = i_k ) 
$ is equal to $M_{ji}$ where  $M$ is the Markov matrix associated to
 $f_l$ we thus have
\[
\pi^{t+1}_i = \sum\limits^{2^l-1}_{j=0}
\pi^t_j. M_{ji} \textrm{ and thus }
\pi^{t+1} = \pi^{t} M.
\]


First of all, 
since the graph $\Gamma(f)$ is strongly connected,
then for all vertices $i$ and $j$, a path can
be  found to  reach $j$  from $i$  in at  most $2^l$  steps.  
There  exists thus $k_{ij} \in \llbracket 1,  2^l \rrbracket$ s.t.
${M}_{ij}^{k_{ij}}>0$.  
As all the multiples $l \times k_{ij}$ of $k_{ij}$ are such that 
${M}_{ij}^{l\times  k_{ij}}>0$, 
we can  conclude that, if
$k$ is the least common multiple of $\{k_{ij}  \big/ i,j  \in \llbracket 1,  2^l \rrbracket  \}$ thus 
$\forall i,j  \in \llbracket  1, 2^l \rrbracket,  {M}_{ij}^{k}>0$ and thus 
$M$ is a regular stochastic matrix.

Let us now recall the following stochastic matrix theorem:
\begin{theorem}[Stochastic Matrix]
  If $M$ is a regular stochastic matrix, then $M$ 
  has an unique stationary  probability vector $\pi$. Moreover, 
  if $\pi^0$ is any initial probability vector and 
  $\pi^{t+1} = \pi^t.M $ for $t = 0, 1,\dots$ then the Markov chain $\pi^t$
  converges to $\pi$ as $t$ tends to infinity.
\end{theorem}

Thanks to this theorem, $M$ 
has an unique stationary  probability vector $\pi$. 
By hypothesis, since $M$ is doubly stochastic we have 
$(\frac{1}{2^l},\dots,\frac{1}{2^l}) = (\frac{1}{2^l},\dots,\frac{1}{2^l})M$
and thus $\pi =  (\frac{1}{2^l},\dots,\frac{1}{2^l})$.
Due to the matrix theorem, there exists some 
$q$ s.t. 
$|\pi^q- \pi| < \epsilon$
and the proof is established.
Since $p(Y| K)$ is $p(X^q)$ the method is then $\epsilon$-stego-secure
provided the strategy-adapter is uniformly distributed.
 \end{proof}

This section has focused on security with regards to probabilistic behaviors. 
Next section studies it in the perspective of topological ones.


\section{Chaos-security}\label{sec:chaossecurity}

 To check whether an existing data hiding scheme is chaotic or not, we propose firstly to write it as an iterate process $x^{n+1}=f(x^n)$. It is possible to prove that this formulation can always be done, as follows. Let us consider a given data hiding algorithm. Because it must be computed one day, it is always possible to translate it as a Turing machine, and this last machine can be written as $x^{n+1} = f(x^n)$ in the following way. Let $(w,i,q)$ be the current configuration of the Turing machine (Fig.~\ref{Turing}), where $w=\sharp^{-\omega} w(0) \hdots w(k)\sharp^{\omega}$ is the paper tape, $i$ is the position of the tape head, $q$ is used for the state of the machine, and $\delta$ is its transition function (the notations used here are well-known and widely used). We define $f$ by:
 \begin{itemize}
 \item $f(w(0) \hdots w(k),i,q) = ( w(0) \hdots w(i-1)aw(i+1)w(k),i+1,q')$, if  $\delta(q,w(i)) = (q',a,\rightarrow)$;
 \item $f( w(0) \hdots w(k),i,q) = (w(0) \hdots w(i-1)aw(i+1)w(k),i-1,q')$,  if $\delta(q,w(i)) = (q',a,\leftarrow)$.
 \end{itemize}
 Thus the Turing machine can be written as an iterate function $x^{n+1}=f(x^n)$ on a well-defined set $\mathcal{X}$, with $x^0$ as the initial configuration of the machine. We denote by $\mathcal{T}(S)$ the iterative process of a data hiding scheme $S$.

 \begin{figure}[h!]
   \centering
 \includegraphics[width=8.5cm]{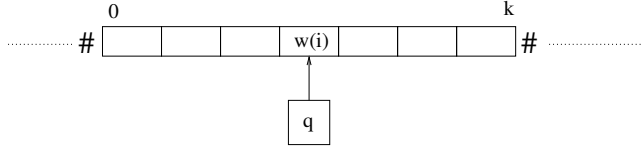}
 \caption{Turing Machine}
 \label{Turing}
 \end{figure}

Let us now define the notion of chaos-security.
Let $\tau$ be a topology on $\mathcal{X}$. So the behavior of this dynamical system can be studied to know whether or not the data hiding scheme is $\tau-$unpredictable. This leads to the following definition.

\begin{definition}
\label{DefinitionChaosSecure}
An information hiding scheme $S$ is said to be chaos-secure on $(\mathcal{X},\tau)$ if its iterative process $\mathcal{T}(S)$ has a chaotic behavior, as defined by Devaney, on this topological space.
\end{definition}

Theoretically speaking, chaos-security can always be studied, as it only requires that the two following points are satisfied.
\begin{itemize}
\item Firstly, the data hiding scheme must be written as an iterate function on
  a set $\mathcal{X}$;
As illustrated by the use of the Turing machine, it is always possible to satisfy this requirement; It is established here since we iterate $G_f$ as defined in
Sect.~(\ref{sub:bdds});

\item Secondly, a metric or a topology must be defined on $\mathcal{X}$; This is always possible, for example, by taking for instance the most relevant one, that is the order topology.
\end{itemize}


Guyeux has recently shown in~\cite{GuyeuxThese10}  that chaotic
iterations of $G_f$ with  the
vectorial negation  as  iterate  function 
have   a  chaotic  behavior.
As a corollary, we deduce that the dhCI  dissimulation algorithm 
with negation mode and  CIIS strategy-adapter is chaos-secure.

However, all these results suffer from only relying on the vectorial
negation function. This problem has been theoretically tackled 
in~\cite{GuyeuxThese10} which provides the  following theorem.

\begin{theorem}
\label{Th:Caracterisation des  IC chaotiques} Functions $f  : \mathds{B}^{n} \to
\mathds{B}^{n}$ such that  $G_f$ is chaotic according to  Devaney, are functions
such that the graph $\Gamma(f)$ is strongly connected.
\end{theorem}
\noindent We deduce from this theorem that functions whose  
graph is strongly connected are sufficient to provide new instances of 
dhCI dissimulation that are chaos-secure.

Computing a mode $f$ such that the image of $n$ (\textit{i.e.}, $f_n$) 
is a function with a strongly connected graph of iterations $\Gamma(f_n)$
has been previously studied (see~\cite{bcgr11:ip} for instance). 
The next section presents a use of them in our steganography context.

\section{Applications to frequential domains}
\label{sec:applications}
We are then left to provide an u.d. strategy-adapter that is independent
from the cover, an image mode $f_l$ whose iteration
graph $\Gamma(f_l)$ is strongly connected and whose Markov
matrix is doubly stochastic.

First, the $\textit{CIIS}(K,y,\alpha,l)$ strategy adapter (see Section~\ref{sub:wmcoding})
has the required properties:
it does not depend on the cover and the proof that its outputs
are u.d. on $\llbracket 1; l \rrbracket$ 
is left as an exercise for the reader.
In all the experiments parameters $K$ and $\alpha$ are randomly 
chosen in $\rrbracket 0, 1\llbracket$ and $\rrbracket 0, 0.5\llbracket$
respectively.
The number of iteration is set to $4*lm$, where $lm$ is the number of LSCs 
that depends on the domain.  
 
Next,~\cite{bcgr11:ip} has presented an iterative approach to generate image
modes $f_l$ such that $\Gamma(f_l)$ is strongly connected. Among these
maps, it is obvious to check which verifies or not the doubly
stochastic constrain.
For instance, in what follows we consider the mode
$f_l: \Bool^l \rightarrow \Bool^l$ s.t. its $i$th component is
defined by
\begin{equation}
\label{eq:fqq}
{f_l}(x)_i =
\left\{
\begin{array}{l}
\overline{x_i} \textrm{ if $i$ is odd} \\
x_i \oplus x_{i-1} \textrm{ if $i$ is even}
\end{array}
\right.
\end{equation}

Thanks to~\cite[Theorem 2]{bcgr11:ip} we deduce that its iteration graph 
$\Gamma(f_l)$ is strongly connected. 
Next, the Markov chain is stochastic by construction. 

Let us prove that its Markov chain is doubly stochastic by induction on the 
length $l$.
For $l=1$ and $l=2$ the proof is obvious. Let us consider that the 
result is established until $l=2k$ for some $k \in \Nats$.

Let us then firstly prove the doubly stochasticity for $l=2k+1$.
Following notations introduced in~\cite{bcgr11:ip}, 
let  $\Gamma(f_{2k+1})^0$ and $\Gamma(f_{2k+1})^1$ denote
the subgraphs of $\Gamma(f_{2k+1})$ induced by the subset $\Bool^{2k} \times\{0\}$
and $\Bool^{2k} \times\{1\}$ of $\Bool^{2k+1}$ respectively.
$\Gamma(f_{2k+1})^0$ and   $\Gamma(f_{2k+1})^1$ are isomorphic to $\Gamma(f_{2k})$.
Furthermore, these two graphs are linked together only with arcs of the form
$(x_1,\dots,x_{2k},0) \to (x_1,\dots,x_{2k},1)$ and 
$(x_1,\dots,x_{2k},1) \to (x_1,\dots,x_{2k},0)$.
In $\Gamma(f_{2k+1})$ the number of arcs whose extremity is $(x_1,\dots,x_{2k},0)$
is  the same than the number of arcs whose extremity is $(x_1,\dots,x_{2k})$ 
augmented with 1, and similarly for $(x_1,\dots,x_{2k},1)$.
By induction hypothesis, the Markov chain associated to $\Gamma(f_{2k})$ is doubly stochastic. All the vertices $(x_1,\dots,x_{2k})$ have thus the same number of 
ingoing arcs and the proof is established for $l$ is $2k+1$.

Let us then  prove the doubly stochasticity for $l=2k+2$.
The map $f_l$ is defined by 
$f_l(x)= (\overline{x_1},x_2 \oplus x_{1},\dots,\overline{x_{2k+1}},x_{2k+2} \oplus x_{2k+1})$.
With previously defined  notations, let us focus on 
$\Gamma(f_{2k+2})^0$ and   $\Gamma(f_{2k+2})^1$ which are isomorphic to $\Gamma(f_{2k+1})$. 
Among configurations of $\Bool^{2k+2}$, only four suffixes of length 2 can be
obviously observed, namely, $00$, $10$, $11$ and $01$.
Since 
$f_{2k+2}(\dots,0,0)_{2k+2}=0$, $f_{2k+2}(\dots,1,0)_{2k+2}=1$, 
$f_{2k+2}(\dots,1,1)_{2k+2}=0$, and $f_{2k+2}(\dots,0,1)_{2k+2}=1$, the number of 
arcs whose extremity is 
\begin{itemize}
\item $(x_1,\dots,x_{2k},0,0)$
 is the same than the one whose extremity is $(x_1,\dots,x_{2k},0)$ in $\Gamma(f_{2k+1})$ augmented with 1 (loop over configurations $(x_1,\dots,x_{2k},0,0)$);
\item $(x_1,\dots,x_{2k},1,0)$
 is the same than the one whose extremity is $(x_1,\dots,x_{2k},0)$ in $\Gamma(f_{2k+1})$ augmented with 1 (arc from configurations 
$(x_1,\dots,x_{2k},1,1)$ to configurations 
$(x_1,\dots,x_{2k},1,0)$);
\item $(x_1,\dots,x_{2k},0,1)$
 is the same than the one whose extremity is $(x_1,\dots,x_{2k},0)$ in $\Gamma(f_{2k+1})$ augmented with 1 (loop over configurations $(x_1,\dots,x_{2k},0,1)$);
\item $(x_1,\dots,x_{2k},1,1)$
 is the same than the one whose extremity is $(x_1,\dots,x_{2k},1)$ in $\Gamma(f_{2k+1})$ augmented with 1 (arc from configurations 
$(x_1,\dots,x_{2k},1,0)$ to configurations 
$(x_1,\dots,x_{2k},1,1)$).
\end{itemize}
Thus all the vertices $(x_1,\dots,x_{2k})$ have  the same number of 
ingoing arcs and the proof is established for $l=2k+2$.

\subsection{DWT embedding}

Let us now explain how the dhCI dissimulation can be applied in
the discrete wavelets domain (DWT).
In this paper, the Daubechies family of wavelets is chosen: 
each DWT decomposition depends on a decomposition level and a coefficient
matrix (Figure~\ref{fig:DWTs}): $\textit{LL}$ means approximation coefficient,
when $\textit{HH},\textit{LH},\textit{HL}$ denote respectively diagonal, 
vertical, and horizontal detail coefficients. 
For example, the DWT coefficient \textit{HH}2 is the matrix equal to the 
diagonal detail coefficient of the second level of decomposition of the image.

\begin{figure}[htb]
\centerline{
\includegraphics[width=7cm]{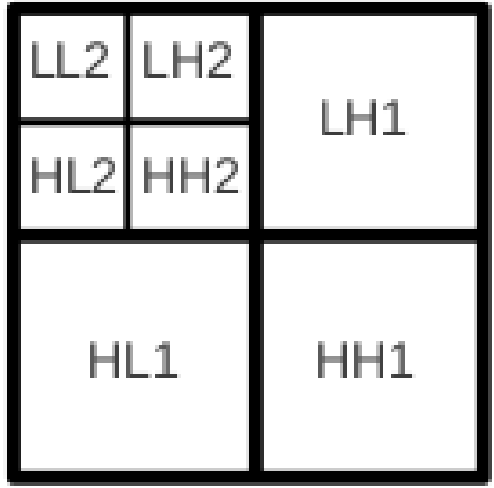}
}
\caption{Wavelets coefficients.}
\label{fig:DWTs}
\end{figure}


The choice of the detail level is motivated by finding
a good compromise between robustness and invisibility.
Choosing low or high frequencies in DWT domain leads either to a very
fragile watermarking without robustness (especially when facing a
JPEG2000 compression attack) or to a large degradation of the host
content. 
In order to have a robust but discrete DWT embedding, 
the second detail level 
(\textit{i.e.}, $\textit{LH}2,\textit{HL}2,\textit{HH}2$) 
that corresponds to the middle frequencies,
has been retained.


Let us consider the Daubechies wavelet coefficients of a third
level decomposition as represented in Figure~\ref{fig:DWTs}. 
We then translate these float coefficients into their 32-bits values.
Let us define the  significance function $u$ that associates to any index $k$ in this sequence of bits the following numbers:
\begin{itemize}
\item $u^k = -1$ if $k$ is one of the three last bits of any index of
  coefficients in  $\textit{LH}2$, $\textit{HL}2$, or in $\textit{HH}2$;
\item $u^k = 0$ if $k$ is an index of a coefficient in  
  $\textit{LH}1$, $\textit{HL}1$, or in $\textit{HH}1$;
\item $u^k = 1$ otherwise.
\end{itemize}

According to the definition of significance of coefficients 
(Def.~\ref{def:msc,lsc}), if $(m,M)$ is $(-0.5,0.5)$,  LSCs are the
last three bits of coefficients in 
$\textit{HL}2$,$\textit{HH}2$, and $\textit{LH}2$.
Thus, decomposition and recomposition functions are fully defined
and dhCI dissimulation scheme can now be applied.

Figure \ref{fig:DWT} shows the result of a
dhCI dissimulation embedding into DWT domain. 
The original is the image 5007 of the BOSS contest~\cite{Boss10}.
Watermark $y$ is given in Fig.~\ref{(b) Watermark}.

From a random selection of 50 images into the database from the BOSS 
contest~\cite{Boss10}, we have applied the previous algorithm with mode $f_l$ 
defined in Equation~(\ref{eq:fqq}) and with the negation mode. 

\begin{figure}[ht]
  \centering
  \subfigure[Original Image.]{\includegraphics[width=0.24\textwidth]
    {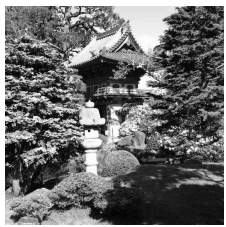}\label{(a) Original 5007}}\hspace{1cm}

  \subfigure[Watermark $y$.]{\includegraphics[width=0.08\textwidth]{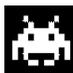}\label{(b) Watermark}}\hspace{1cm}

  \subfigure[Watermarked Image.]{\includegraphics[width=0.24\textwidth]{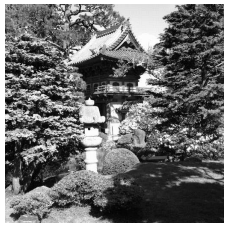}\label{(c) Watermarked 5007}}

\caption{Data hiding in DWT domain}
\label{fig:DWT}
\end{figure}



\subsection{DCT embedding}
Let us denote by $x$ the original image of size $H \times L$, and by $y$
the hidden message, supposed here to be a binary image of size $H' \times L'$. %
The image $x$ is transformed from the spatial
domain to DCT domain frequency bands,
in order to embed $y$ inside it.  
To do so, the host image is firstly divided into $8 \times 8$
image blocks as given below:
$$x = \bigcup_{k=1}^{H/8} \bigcup_{k'=1}^{L/8} x(k,k').$$
Thus, for each image block,
a DCT is performed and the coefficients in the frequency bands 
are obtained as follows:
$x_{DCT}(m;n) = DCT(x(m;n))$.

To define a discrete but robust scheme, only the four following coefficients of each $8 \times 8$ block in position $(m,n)$ will be possibly modified: $x_{DCT}(m;n)_{(3,1)},$ $x_{DCT}(m;n)_{(2,2)},$ or $x_{DCT}(m;n)_{(1,3)}$. 
This choice can be reformulated as follows. 
Coefficients of each DCT matrix are re-indexed by using a southwest/northeast diagonal, such that $i_{DCT}(m,n)_1 = x_{DCT}(m;n)_{(1,1)}$,\linebreak $i_{DCT}(m,n)_2 = x_{DCT}(m;n)_{(2,1)}$, $i_{DCT}(m,n)_3 = x_{DCT}(m;n)_{(1,2)}$, $i_{DCT}(m,n)_4 = x_{DCT}(m;n)_{(3,1)}$, ..., and $i_{DCT}(m,n)_{64} =$  $ x_{DCT}(m;n)_{(8,8)}$.
So the signification function can be defined in this context by:
\begin{itemize}
\item if $k$ mod $64 \in \{1,2,3\}$ and $k\leqslant H\times L$, then $u^k=1$;
\item else if $k$ mod $64 \in \{4, 5, 6\}$ and $k\leqslant H\times L$, then $u^k=-1$;
\item else $u^k = 0$.
\end{itemize}
The significance of coefficients are obtained for instance with 
$(m,M)=(-0.5,0.5)$ leading to the definitions of MSCs, LSCs, and passive coefficients.
Thus, decomposition and recomposition functions are fully defined and dhCI dissimulation scheme can now be applied.

\subsection{Image quality}
This section focuses on measuring visual quality of our steganographic method.
Traditionally, this is achieved by quantifying the similarity 
between the modified image and its reference image.
The Mean Squared Error (MSE) and the Peak Signal to Noise
Ratio (PSNR) are the most widely known tools that provide such a metric.
However, both of them do not take into account Human Visual System (HVS)
properties. 
Recent works~\cite{EAPLBC06,SheikhB06,PSECAL07,MB10} have tackled this problem 
by creating new metrics. Among them, what follows focuses on PSNR-HVS-M~\cite{PSECAL07} and BIQI~\cite{MB10}, considered as advanced visual quality metrics.  
The former efficiently combines PSNR and  visual between-coefficient contrast masking of DCT basis functions based on HVS. This metric has 
been computed here by using the implementation given at~\cite{psnrhvsm11}.
The latter allows to get a blind image quality assessment measure, 
\textit{i.e.}, without any knowledge of the source distortion.
Its implementation is available at~\cite{biqi11}.

\begin{table}
\begin{center}
\begin{tabular}{|c|c|c|c|c|}
\hline
Embedding & \multicolumn{2}{|c|}{DWT} 
 & \multicolumn{2}{|c|}{DCT} \\
\hline
Mode & $f_l$ & neg. & $f_l$ & neg. \\
\hline
PSNR & 42.74     & 42.76     &  52.68      &  52.41   \\
\hline
PSNR-HVS-M & 44.28  & 43.97 & 45.30 & 44.93 \\
\hline
BIQI & 35.35 & 32.78 & 41.59 & 47.47 \\
\hline
\end{tabular}
\end{center}
\caption{Quality measeures of our steganography approach\label{table:quality}} 
\end{table}

Results of the image quality metrics 
are summarized into the Table~\ref{table:quality}.
In wavelet domain, the PSNR values obtained here are comparable to other approaches
(for instance, PSNR are 44.2 in~\cite{TCL05} and 46.5 in~\cite{DA10}), 
but  a real improvement for the discrete cosine embeddings is obtained 
(PSNR is 45.17 for~\cite{CFS08}, it is always lower than 48 for~\cite{Mohanty:2008:IWB:1413862.1413865}, and always lower than 39 for~\cite{MK08}).
Among steganography approaches that evaluate PSNR-HVS-M, results of our approach 
are convincing. Firstly, optimized method developed along~\cite{Randall11} has a PSNR-HVS-M equal to 44.5 whereas our approach, with a similar PSNR-HVS-M, should be easily improved by considering optimized mode. Next, 
another approach~\cite{Muzzarelli:2010} have higher PSNR-HVS-M, certainly, but
this work does not address robustness evaluation whereas our approach is complete.
Finally, as far as we know, this work is the first one that evaluates the BIQI metric in the steganography context.

With all this material, we are then left to evaluate the robustness of this 
approach.

\subsection{Robustness}
Previous sections have formalized frequential domains embeddings and
has focused on the negation mode and $f_l$ defined in Equ.~(\ref{eq:fqq}).
In the robustness given in this continuation, {dwt}(neg), 
{dwt}(fl), {dct}(neg), {dct}(fl) 
respectively stand for the DWT and DCT embedding 
with the negation mode and with this instantiated mode.
 
For each experiment, a set of 50 images is randomly extracted 
from the database taken from the BOSS contest~\cite{Boss10}. 
Each cover is a $512\times 512$ grayscale digital image and the watermark $y$ 
is given in Fig~\ref{(b) Watermark}. 
Testing the robustness of the approach is achieved by successively applying
on watermarked images attacks like cropping, compression, and geometric 
transformations.
Differences between 
$\hat{y}$ and $\varphi_m(z)$ are 
computed. Behind a given threshold rate, the image is said to be watermarked.  
Finally, discussion on metric quality of the approach is given in 
Sect.~\ref{sub:roc}.

Robustness of the approach is  evaluated by
applying different percentage of cropping: from 1\% to 81\%.
Results are presented in Fig.~\ref{Fig:atck:dec}. Fig.~\ref{Fig:atq:dec:img}
gives the cropped image 
where 36\% of the image is removed.
Fig.~\ref{Fig:atq:dec:curves} presents effects of such an attack.
From this experiment, one can conclude that all embeddings have similar 
behaviors.
All the percentage differences are so far less than 50\% 
(which is the mean random error) and thus robustness is established.

\begin{figure}[ht]
  \centering
  \subfigure[Cropped Image.]{\includegraphics[width=0.24\textwidth]
    {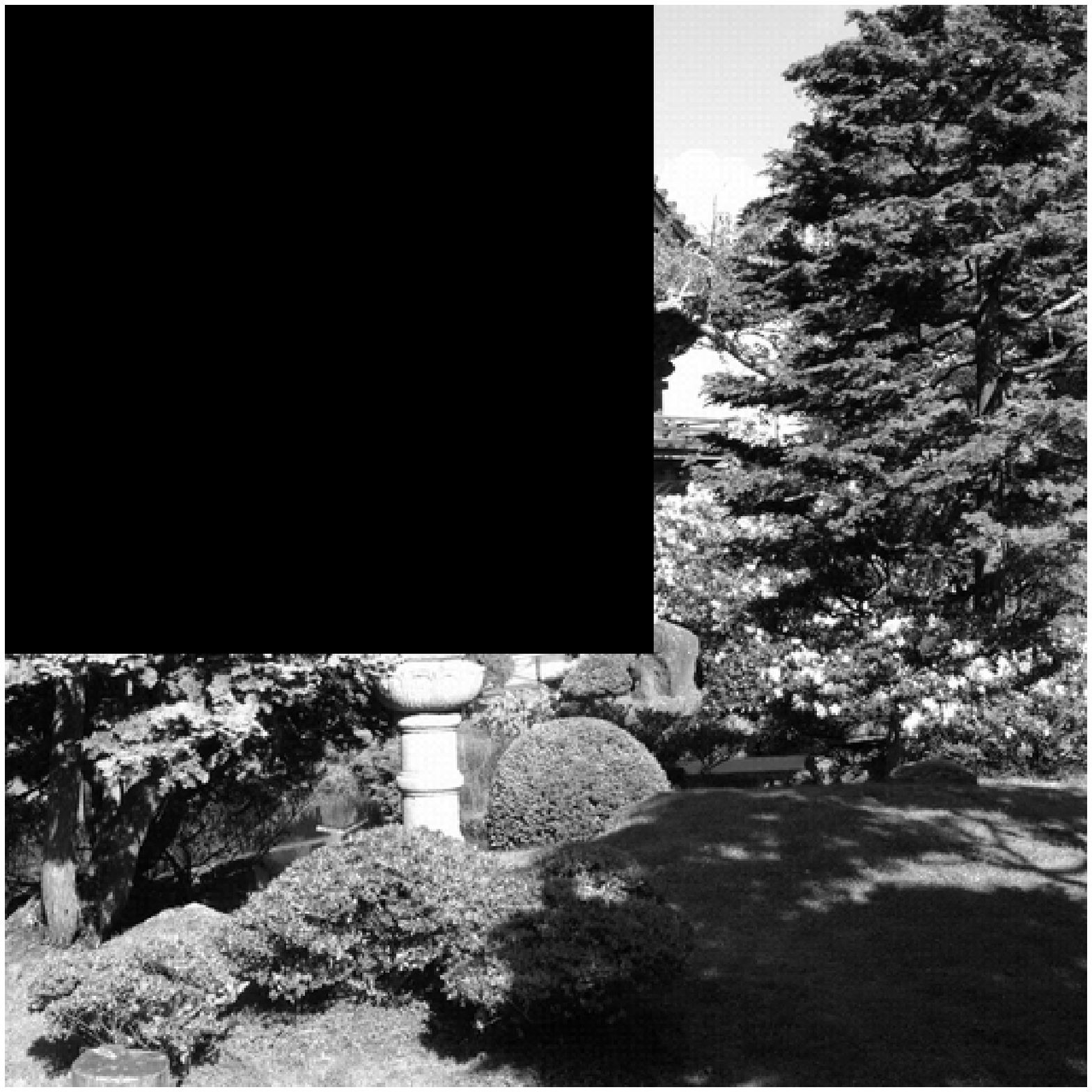}\label{Fig:atq:dec:img}}\hspace{2cm}
  \subfigure[Cropping Effect]{
\includegraphics[width=0.5\textwidth]{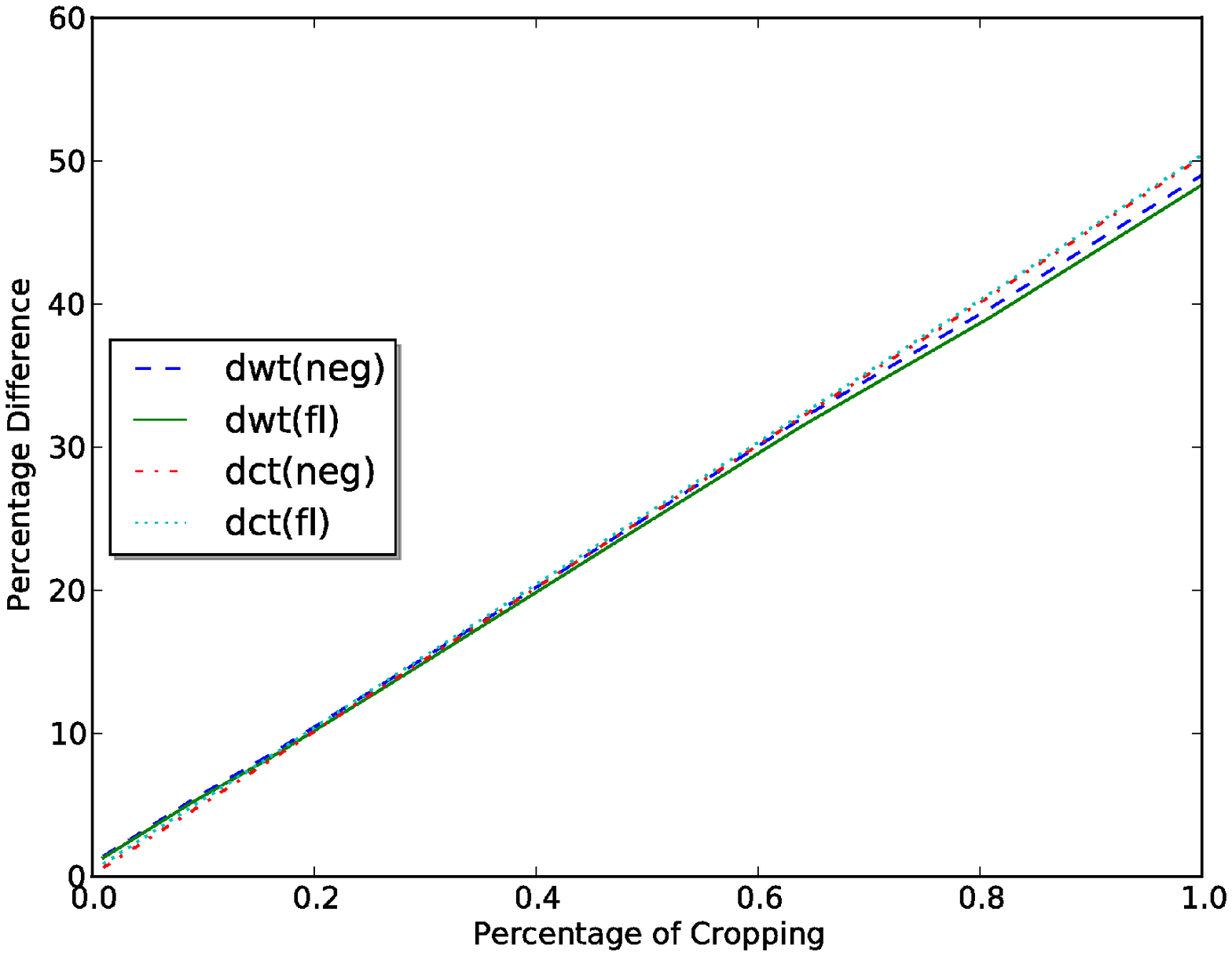}\label{Fig:atq:dec:curves}}
\caption{Cropping Results}
\label{Fig:atck:dec}
\end{figure}

\subsubsection{Robustness against compression}

Robustness against compression is addressed
by studying both JPEG  and JPEG 2000 image compressions.
Results are respectively presented in Fig.~\ref{Fig:atq:jpg:curves}
and Fig.~\ref{Fig:atq:jp2:curves}.
Without surprise, DCT embedding which is based on DCT 
(as JPEG compression algorithm is) is  more 
adapted to JPEG compression than DWT embedding.
Furthermore, we have a similar behavior for the JPEG 2000 compression algorithm, which is based on wavelet encoding: DWT embedding naturally outperforms
DCT one in that case.

\begin{figure}[ht]
  \centering
  \subfigure[JPEG Effect]{
\includegraphics[width=0.45\textwidth]{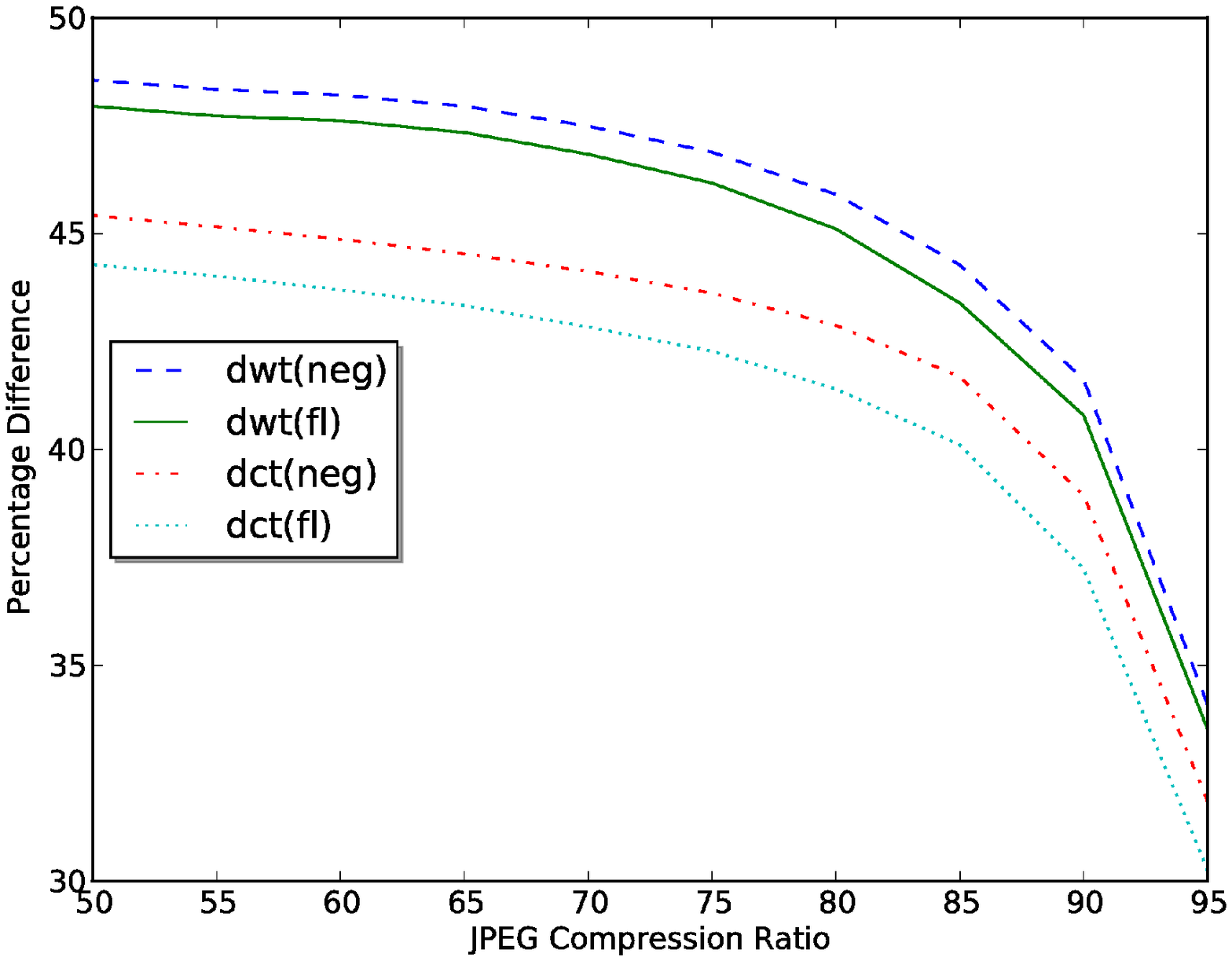}\label{Fig:atq:jpg:curves}}
  \subfigure[JPEG 2000 Effect]{
\includegraphics[width=0.45\textwidth]{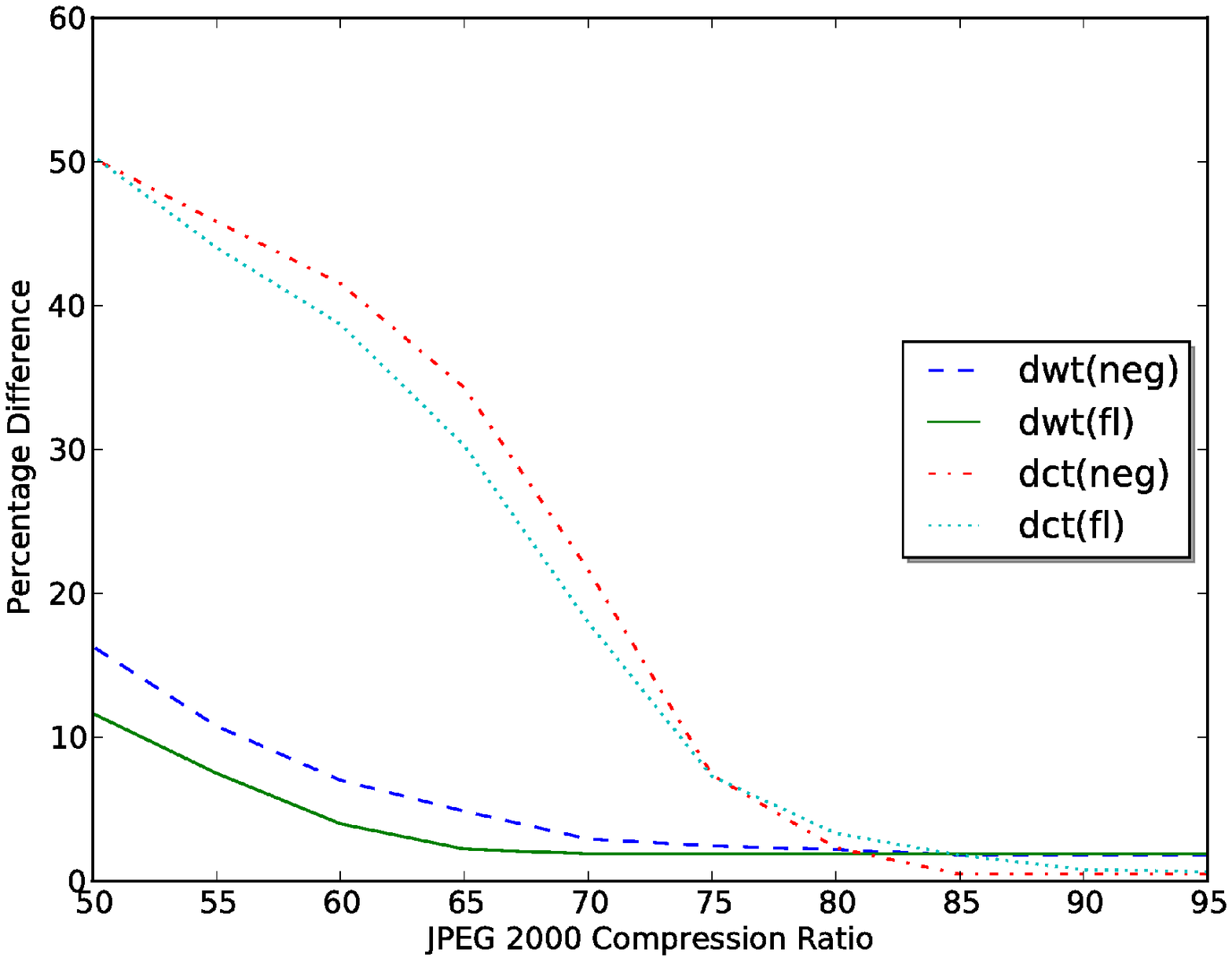}\label{Fig:atq:jp2:curves}}
\caption{Compression Results}
\label{Fig:atck:comp}
\end{figure}

\subsubsection{Robustness against Contrast and Sharpness Attack}
Contrast and Sharpness adjustment belong to the the classical set of 
filtering image attacks.
Results of such attacks are presented in 
Fig.~\ref{Fig:atq:fil} where 
Fig.~\ref{Fig:atq:cont:curve} and Fig.~\ref{Fig:atq:sh:curve} summarize 
effects of contrast and sharpness adjustment respectively.

\begin{figure}[ht]
  \centering
  \subfigure[Contrast Effect]{
\includegraphics[width=0.45\textwidth]{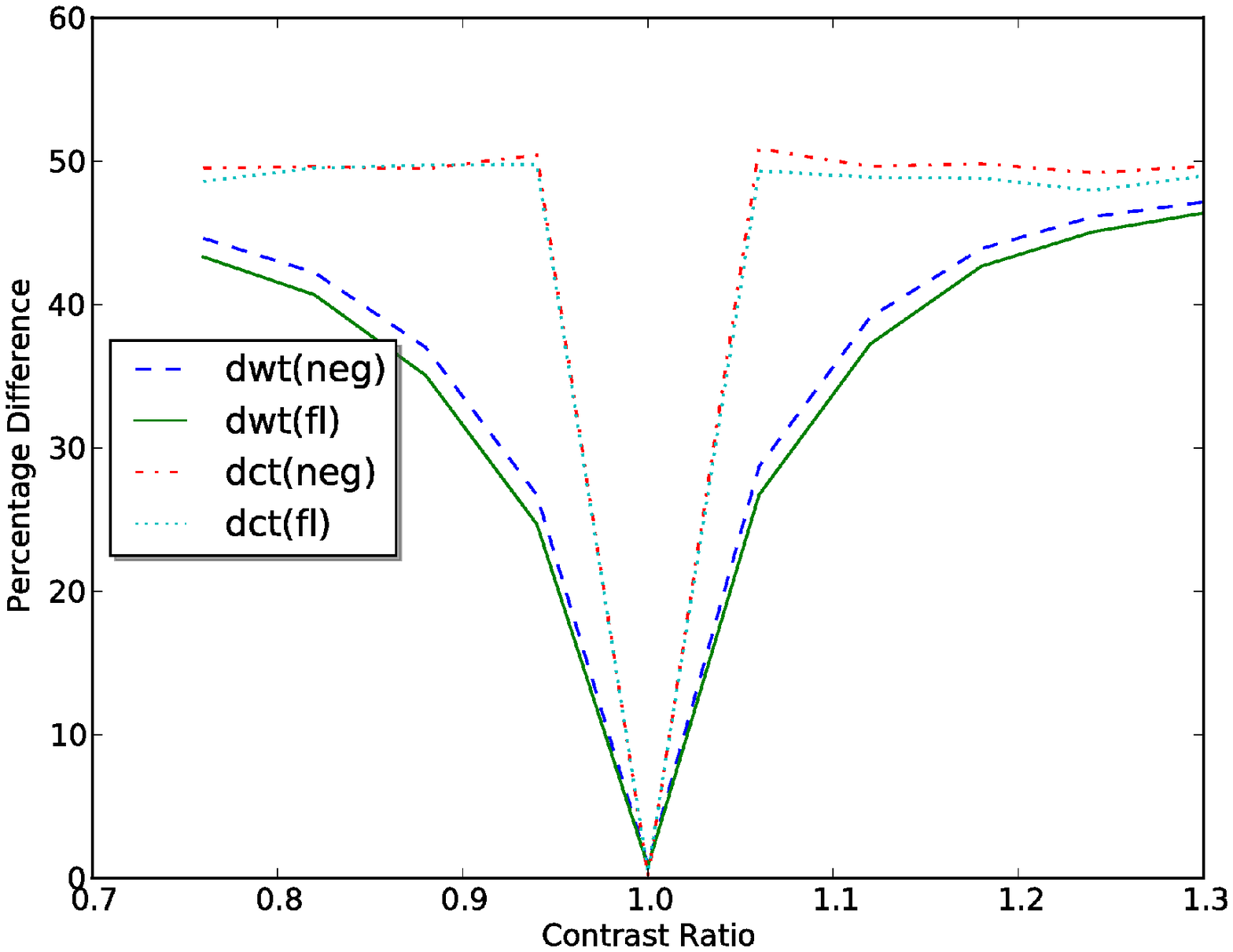}\label{Fig:atq:cont:curve}}
  \subfigure[Sharpness Effect]{
\includegraphics[width=0.45\textwidth]{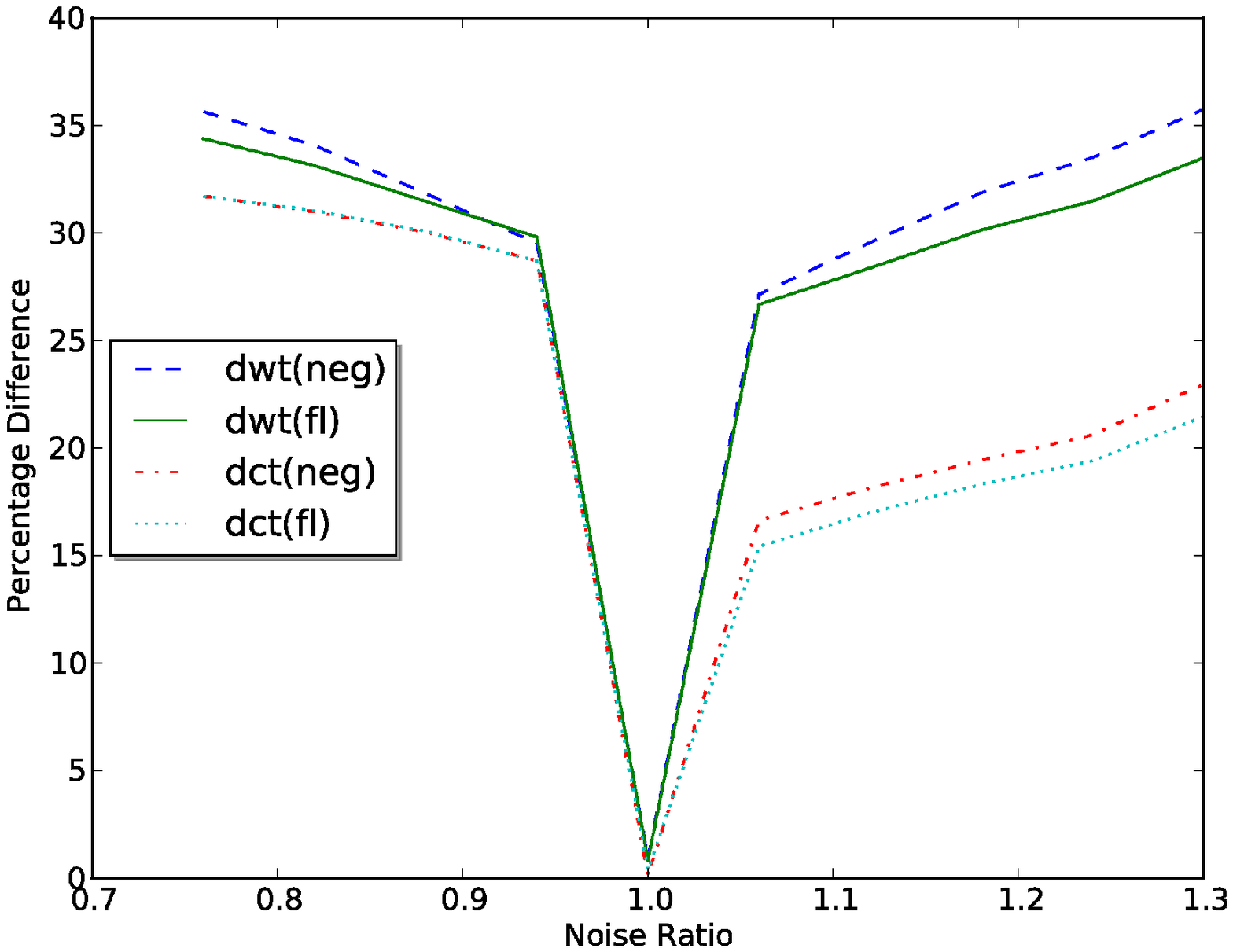}\label{Fig:atq:sh:curve}}
\caption{Filtering Results}
\label{Fig:atq:fil}
\end{figure}

\subsubsection{Robustness against Geometric Transformation}
Among geometric transformations, we focus on  
rotations, \textit{i.e.}, when two opposite rotations 
of angle $\theta$ are successively applied around the center of the image.
In these geometric transformations,  angles range from 2 to 20 
degrees.  
Results are presented in Fig.~\ref{Fig:atq:rot}: Fig.~\ref{Fig:atq:rot:img}
gives the image of a rotation of 20 degrees whereas
Fig.~\ref{Fig:atq:rot:curve} presents effects of such an attack.
It is not a surprise that results are better for DCT embeddings: this approach
is based on cosine as rotation is.

\begin{figure}[ht]
  \centering
  \subfigure[20 degrees Rotation Image]{
    \includegraphics[width=0.25\textwidth]{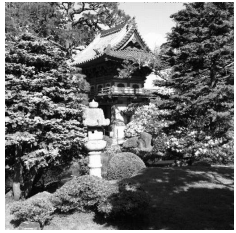}\label{Fig:atq:rot:img}}

  \subfigure[Rotation Effect]{
\includegraphics[width=0.45\textwidth]{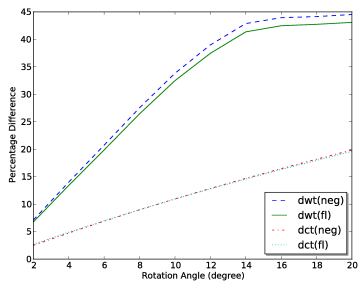}\label{Fig:atq:rot:curve}}

\caption{Rotation Attack Results}
\label{Fig:atq:rot}
\end{figure}

\subsection{Evaluation of the  Embeddings}\label{sub:roc}
We are then left to set a convenient threshold that is accurate to 
determine whether an image is watermarked or not.
Starting from a set of 100 images selected among the Boss image Panel,
we compute the following three sets: 
the one with all the watermarked images $W$,
the one with all successively watermarked and attacked images $\textit{WA}$,
and the one with only the attacked images $A$.
Notice that the 100 attacks for each images 
are selected among these detailed previously.

For each threshold $t \in \llbracket 0,55 \rrbracket$ and a given image 
$x \in \textit{WA} \cup A$, 
differences on DCT are computed. The image is said to be watermarked
if these differences are less than the threshold. 
In the positive case and if $x$ really belongs to 
$\textit{WA}$ it is a True Positive (TP) case.  
In the negative case but if $x$ belongs to 
$\textit{WA}$ it is a False Negative (FN) case.  
In the positive case but if $x$  belongs to 
$\textit{A}$, it is a False Positive (FP) case.  
Finally, in the negative case and if $x$ belongs to
$\textit{A}$, it is a True Negative (TN).  
The True (resp. False) Positive Rate  (TPR) (resp. FPR) is thus computed 
by dividing the number of TP (resp. FP) by 100.

\begin{figure}[ht]
\begin{center}
\includegraphics[width=7cm]{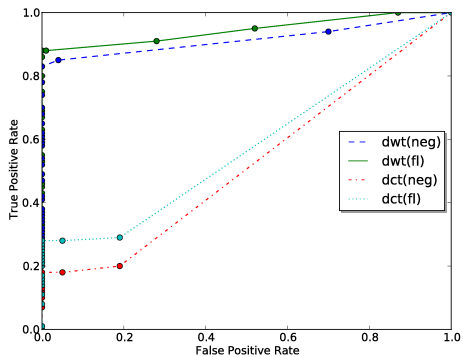}
\end{center}
\caption{ROC Curves for DWT or DCT Embeddings}\label{fig:roc:dwt}
\end{figure}

The Figure~\ref{fig:roc:dwt} is the Receiver Operating Characteristic (ROC) 
curve. 
For the DWT, it shows that best results are obtained when the threshold 
is 45\% for the dedicated function (corresponding to the point (0.01, 0.88))
and  46\% for the negation function (corresponding to the point (0.04, 0.85)).
It allows to conclude that each time LSCs differences between
a watermarked image and another given image  $i'$ are less than 45\%, we can claim that 
$i'$ is an attacked version of the original watermarked content.
For the two DCT embeddings, best results are obtained when the threshold 
is 44\% (corresponding to the points (0.05, 0.18) and (0.05, 0.28)).

Let us then give some confidence intervals for all the evaluated attacks. The
approach is resistant to:
\begin{itemize}
\item all the croppings where percentage is less than 85;
\item compressions where quality ratio is greater 
  than 82 with DWT embedding and 
  where quality ratio is greater than 67 with DCT one;
\item contrast when strengthening belongs to $[0.76,1.2]$ 
(resp. $[0.96,1.05]$)  in DWT (resp. in DCT) embedding;
\item all the rotation attacks with DCT embedding and a rotation where
angle is less than 13 degrees with DWT one.
\end{itemize}

\section{Conclusion}\label{sec:concl}
This paper has proposed a new class of secure and robust information hiding 
algorithms.
It has been entirely formalized, thus allowing both its theoretical security 
analysis, and the computation of numerous variants encompassing spatial and 
frequency domain embedding.
After having presented the general algorithm with detail, we have given
conditions for choosing mode and strategy-adapter making the whole
class  stego-secure or $\epsilon$-stego-secure.
To our knowledge, this is the first time such a result has been established.

Applications in frequency domains (namely DWT and DCT domains) have finally be
formalized.
Complete experiments have allowed us 
first to evaluate how invisible is the steganographic method (thanks to the PSNR computation) and next to verify the robustness property against attacks.
Furthermore, the use of ROC curves for DWT embedding have revealed very high rates
between True positive and False positive results.

In future work, our intention is to find the best image mode with respect to  
the combination between  DCT and DWT based steganography
algorithm. Such a combination topic has already been addressed
(\textit{e.g.}, in~\cite{al2007combined}), but never with objectives
we have set.

Additionally, we will try to discover new topological properties for the dhCI
dissimulation schemes.
Consequences of these chaos properties will be drawn in the context of 
information hiding security.
We will especially focus on the links between topological properties and classes
of attacks, such as KOA, KMA, EOA, or CMA.

Moreover, these algorithms will be compared to other existing ones, among other
things by testing whether these algorithms are chaotic or not.
Finally we plan to verify the robustness of our approach 
against statistical steganalysis methods~\cite{GFH06,ChenS08,DongT08,FridrichKHG11a}.

\bibliographystyle{compj}

\bibliography{abbrev2,mabase,biblioand2}

\begin{thebibliography}{99}

\bibitem{Houmansadr09}
Houmansadr, A., Kiyavash, N., and Borisov, N. (2009) Rainbow: A robust and
  invisible non-blind watermark for network flows.
\newblock {\em Proceedings of the Network and Distributed System Security
  Symposium, NDSS 2009,San Diego},  February. The Internet Society,  Washington
  DC/Reston.

\bibitem{P1150442004}
G.S.El-Taweel, Onsi, H., M.Samy, and Darwish, M. (2005) Secure and non-blind
  watermarking scheme for color images based on dwt.
\newblock {\em ICGST International Journal on Graphics, Vision and Image
  Processing}, {\bf  05}, 1--5.

\bibitem{Wu2007}
Wu, X., Guan, Z.-H., and Wu, Z. (2007) A chaos based robust spatial domain
  watermarking algorithm.
\newblock {\em ISNN '07: Proceedings of the 4th international symposium on
  Neural Networks, Nanjing, China},  Berlin,  June,  Lecture Notes in Computer
  Science, {\bf4492},  pp. 113--119. Springer-Verlag.

\bibitem{Liu07}
Liu, Z. and Xi, L. (2007) Image information hiding encryption using chaotic
  sequence.
\newblock {\em Knowledge-Based Intelligent Information and Engineering Systems,
  11th International Conference, KES 2007, XVII Italian Workshop on Neural
  Networks, Vietri sul Mare, Italy},  Berlin,  September,  pp. 202--208.
  Springer-Verlag.

\bibitem{CongJQZ06}
Cong, J., Jiang, Y., Qu, Z., and Zhang, Z. (2006) A wavelet packets
  watermarking algorithm based on chaos encryption.
\newblock {\em Computational Science and Its Applications - ICCSA 2006,
  International Conference, Glasgow, UK},  Berlin,  May,  Lecture Notes in
  Computer Science, {\bf3980},  pp. 921--928. Springer-Verlag.

\bibitem{Zhu06}
Congxu, Z., Xuefeng, L., and Zhihua, L. (2006) Chaos-based multipurpose image
  watermarking algorithm.
\newblock {\em Wuhan University Journal of Natural Sciences}, {\bf  11},
  1675--1678.

\bibitem{Wu2007bis}
Wu, X. and Guan, Z.-H. (2007) A novel digital watermark algorithm based on
  chaotic maps.
\newblock {\em Physics Letters A}, {\bf  365}, 403 -- 406.

\bibitem{Cayre2008}
Cayre, F. and Bas, P. (2008) Kerckhoffs-based embedding security classes for
  woa data hiding.
\newblock {\em IEEE Transactions on Information Forensics and Security}, {\bf
  3}, 1--15.

\bibitem{Cox97securespread}
Cox, I.~J., Member, S., Kilian, J., Leighton, F.~T., and Shamoon, T. (1997)
  Secure spread spectrum watermarking for multimedia.
\newblock {\em IEEE Transactions on Image Processing}, {\bf  6}, 1673--1687.

\bibitem{gfb10:ip}
Guyeux, C., Friot, N., and Bahi, J. (2010) Chaotic iterations versus
  spread-spectrum: chaos and stego security.
\newblock {\em Sixth International Conference on Intelligent Information Hiding
  and Multimedia Signal Processing (IIH-MSP 2010), Darmstadt, Germany},
  Washington, DC,  October,  pp. 208--211. IEEE Computer Society.

\bibitem{GuyeuxThese10}
Guyeux, C. (2010) Le d\'{e}sordre des it\'{e}rations chaotiques et leur
  utilit\'{e} en sécurit\'{e} informatique.
\newblock PhD thesis Universit\'{e} de Franche-Comt\'{e}.

\bibitem{Devaney}
Devaney, R.~L. (2003) {\em An Introduction to Chaotic Dynamical Systems, 2nd
  Edition}. Westview Press,  Boulder, CO.

\bibitem{Shujun1}
Shujun, L., Qi, L., Wenmin, L., Xuanqin, M., and Yuanlong, C. (2001)
  Statistical properties of digital piecewise linear chaotic maps and their
  roles in cryptography and pseudo-random coding.
\newblock {\em Cryptography and Coding, 8th IMA International Conference,
  Cirencester, UK},  Berlin,  December,  Lecture Notes in Computer Science,
  {\bf2260},  pp. 205--221. Springer-Verlag.

\bibitem{Arroyo08}
Arroyo, D., Alvarez, G., and Fernandez, V. (2008) On the inadequacy of the
  logistic map for cryptographic applications.
\newblock arXiv/0805.4355.

\bibitem{Cachin2004}
Cachin, C. (1998) An information-theoretic model for steganography.
\newblock {\em Information Hiding, Second International Workshop, Portland,
  Oregon, USA},  Berlin,  April,  Lecture Notes in Computer Science, {\bf1525},
   pp. 306--318. Springer-Verlag.

\bibitem{Mittelholzer99}
Mittelholzer, T. (1999) An information-theoretic approach to steganography and
  watermarking.
\newblock {\em nformation Hiding, Third International Workshop, IH'99, Dresden,
  Germany},  Berlin,  September,  pp. 1--16. Springer-Verlag.

\bibitem{Kalker2001}
Kalker, T. (2001) Considerations on watermarking security.
\newblock {\em 2001 IEEE Fourth Workshop on Multimedia Signal Processing,
  Cannes, France},  Washington, DC,  October,  pp. 201--206. IEEE Computer
  Society.

\bibitem{Furon2002}
Furon, T. (2002).
\newblock Security analysis.
\newblock European Project IST-1999-10987 CERTIMARK, Deliverable D.5.5.

\bibitem{Cayre2005}
Cayre, F., Fontaine, C., and Furon, T. (2005) Watermarking security: theory and
  practice.
\newblock {\em IEEE Transactions on Signal Processing}, {\bf  53}, 3976--3987.

\bibitem{Perez06}
Perez-Freire, L., Pérez-Gonzalez, F., and Comesaña, P. (2006) Secret dither
  estimation in lattice-quantization data hiding: A set-membership approach.
\newblock {\em Security, Steganography, and Watermarking of Multimedia
  Contents, San Jose, California},  Bellingham, WA,  January,  pp. 1--12.
  Society of Photo-Optical Instrumentation Engineers.

\bibitem{Simmons83}
Simmons, G.~J. (1984) The prisoners' problem and the subliminal channel.
\newblock {\em Advances in Cryptology, Proc. CRYPTO'83, University of
  California, Santa Barbara},  New York,  August,  pp. 51--67. Plenum Press.

\bibitem{bcgr11:ip}
Bahi, J., Couchot, J.-F., Guyeux, C., and Richard, A. (2011) On the link
  between strongly connected iteration graphs and chaotic boolean discrete-time
  dynamical systems.
\newblock {\em FCT'11, 18th Int. Symp. on Fundamentals of Computation
  Theory,Oslo, Norway},  Berlin,  August,  Lecture Notes in Computer Science,
  {\bf6914},  pp. 126--137. Springer-Verlag.

\bibitem{Boss10}
Pevný, T., Filler, T., and Bas, P. (2010).
\newblock Break our steganographic system.
\newblock available at \url{http://www.agents.cz/boss/}.

\bibitem{EAPLBC06}
Egiazarian, K., Astola, J., Ponomarenko, V., Nikolayand~Lukin, Battisti, F.,
  and Carli, M. (2006) New full-reference quality metrics based on hvs.
\newblock In Li, B. (ed.), {\em CD-ROM Proceedings of the Second International
  Workshop on Video Processing and Quality Metrics, Scottsdale, USA},  January.

\bibitem{SheikhB06}
Sheikh, H.~R. and Bovik, A.~C. (2006) Image information and visual quality.
\newblock {\em IEEE Transactions on Image Processing}, {\bf  15}, 430--444.

\bibitem{PSECAL07}
Ponomarenko, N., Silvestri, F., Egiazarian, K., Carli, M., Astola, J., and
  Lukin, V. (2007) On between-coefficient contrast masking of dct basis
  functions.
\newblock In Li, B. (ed.), {\em CD-ROM Proceedings of the Third International
  Workshop on Video Processing and Quality Metrics for Consumer Electronics
  VPQM-07,Scottsdale, Arizona, USA},  January.

\bibitem{MB10}
Moorthy, A. K.~M. and Bovik, A.~C. (2010) A two-step framework for constructing
  blind image quality indices.
\newblock {\em IEEE Signal Processing Letters}, {\bf  17}, 513--516.

\bibitem{psnrhvsm11}
(2011).
\newblock Psnr-hvs-m page.
\newblock \url{http://www.ponomarenko.info/psnrhvsm.htm}.

\bibitem{biqi11}
(2011).
\newblock Biqi page.
\newblock \url{http://live.ece.utexas.edu/research/quality/BIQI_release.zip}.

\bibitem{TCL05}
Temi, C., Choomchuay, S., and Lasakul, A. (2005) A robust image watermarking
  using multiresolution analysis of wavelet.
\newblock {\em ISCIT 2005. IEEE International Symposium on Communications and
  Information Technology},  Washington, DC,  October,  pp. 623--626. IEEE
  Computer Society.

\bibitem{DA10}
V.~Dharwadkar, N. and B.B., A. (2010) Watermarking scheme for color images
  using wavelet transform based texture properties and secret sharing.
\newblock {\em International Journal of Information and Communication
  Engineering}, {\bf  6}, 93--100.

\bibitem{CFS08}
Chrysochos, E., Fotopoulos, V., and Skodras, A.~N. (2008) Robust watermarking
  of digital images based on chaotic mapping and dct.
\newblock {\em 16th European Signal Processing Conference (EUSIPCO 2008),
  Lausanne, Switzerland,},  August,  pp. 17--21. EURASIP.

\bibitem{Mohanty:2008:IWB:1413862.1413865}
Mohanty, S.~P. and Bhargava, B.~K. (2008) Invisible watermarking based on
  creation and robust insertion-extraction of image adaptive watermarks.
\newblock {\em ACM Trans. Multimedia Comput. Commun. Appl.}, {\bf  5},
  12:1--12:22.

\bibitem{MK08}
Mohan, B. and Kumar, S. (2008) Robust digital watermarking scheme using
  contourlet transform.
\newblock {\em IJCSNS International Journal of Computer Science and Network
  Security}, {\bf  8}.

\bibitem{Randall11}
Randall, A. (2011) A novel semi-fragile watermarking scheme with iterative
  restoration.
\newblock Available at
  \url{http://www.aaronrandall.com/Files/WatermarkingPaperLight.pdf}.

\bibitem{Muzzarelli:2010}
Muzzarelli, M., Carli, M., Boato, G., and Egiazarian, K. (2010) Reversible
  watermarking via histogram shifting and least square optimization.
\newblock {\em Proceedings of the 12th ACM workshop on Multimedia and security,
  MM\&Sec'10,Roma, Italy},  New York, NY, USA,  pp. 147--152. ACM.

\bibitem{al2007combined}
Al-Haj, A. (2007) Combined dwt-dct digital image watermarking.
\newblock {\em Journal of computer science}, {\bf  3}, 740--746.

\bibitem{GFH06}
Goljan, M., Fridrich, J.~J., and Holotyak, T. (2006) New blind steganalysis and
  its implications.
\newblock {\em Proc. SPIE, Electronic Imaging, Security, Steganography, and
  Watermarking of Multimedia Contents VIII, San Jose. CA.},  Bellingham, WA,
  January. Society of Photo-Optical Instrumentation Engineers.

\bibitem{ChenS08}
Chen, C. and Shi, Y.~Q. (2008) Jpeg image steganalysis utilizing both
  intrablock and interblock correlations.
\newblock {\em International Symposium on Circuits and Systems (ISCAS 2008),
  Seattle, Washington, USA},  Washington, DC,  May,  pp. 3029--3032. IEEE
  Computer Society.

\bibitem{DongT08}
Dong, J. and Tan, T. (2008) Blind image steganalysis based on run-length
  histogram analysis.
\newblock {\em Proceedings of the International Conference on Image Processing,
  ICIP 2008, San Diego, California, USA},  Washington, DC,  October,  pp.
  2064--2067. IEEE Computer Society.

\bibitem{FridrichKHG11a}
Fridrich, J.~J., Kodovsk{\'y}, J., Holub, V., and Goljan, M. (2011)
  Steganalysis of content-adaptive steganography in spatial domain.
\newblock {\em Information Hiding - 13th International Conference, IH 2011,
  Prague, Czech Republic},  Berlin,  May Lecture Notes in Computer Science,
  pp. 102--117. Springer-Verlag.

\end{thebibliography}

\end{document}